\documentclass{article}

\usepackage{arxiv}

\usepackage[utf8]{inputenc}
\usepackage[T1]{fontenc}
\usepackage{hyperref}
\usepackage{url}
\usepackage{booktabs}
\usepackage{amsfonts}
\usepackage{amsmath}
\usepackage{amssymb}
\usepackage{nicefrac}
\usepackage{microtype}
\usepackage{graphicx}
\usepackage{natbib}
\usepackage{doi}
\usepackage{amsthm}
\usepackage{tikz}
\usepackage{pgf}
\usetikzlibrary{arrows.meta, shapes.geometric, positioning, fit,
                decorations.pathreplacing, calc, backgrounds}
\usepackage{listings}
\usepackage{xcolor}
\usepackage{array}
\usepackage{multirow}
\usepackage{tabularx}
\usepackage{caption}
\usepackage{subcaption}
\usepackage{float}
\usepackage{mathtools}

\newtheorem{theorem}{Theorem}[section]
\newtheorem{corollary}[theorem]{Corollary}

\theoremstyle{definition}
\newtheorem{definition}[theorem]{Definition}
\theoremstyle{remark}

\definecolor{isalblue}{RGB}{0,70,140}
\definecolor{isalgreen}{RGB}{30,110,50}
\definecolor{isalgray}{RGB}{110,110,110}
\definecolor{isalred}{RGB}{160,30,30}
\definecolor{isalbg}{RGB}{248,248,252}

\lstdefinelanguage{IsalProgram}{
  morekeywords={J,H,W,Np,Ns,Nt,Pp,Ps,Pt,Nj,Pj,D,
                Mps,Mpt,Msp,Mst,Mtp,Mts,Mji,
                Ib,Ii,If,Is,
                Cps,Cpt,Csp,Cst,Ctp,Cts,
                Aa,As,Am,Ad,An,Aq,
                Sc,Sx,
                Zp,Zs,Zt,Le,Lp,
                Bp,Bs,Bt,
                L1,L2,L3,L4,L5,L6,L7,L8,L9,L10,L11,L12,L13,L14,L15},
  sensitive=true,
  morecomment=[l]{;},
  morestring=[b]",
  basicstyle=\ttfamily\small,
  keywordstyle=\color{isalblue}\bfseries,
  commentstyle=\color{isalgray}\itshape,
  stringstyle=\color{isalred},
  backgroundcolor=\color{isalbg},
  frame=single,
  framerule=0.4pt,
  rulecolor=\color{isalgray!50},
  numbers=left,
  numberstyle=\tiny\color{isalgray},
  stepnumber=1,
  numbersep=6pt,
  tabsize=4,
  showspaces=false,
  showstringspaces=false,
  breaklines=true,
  captionpos=b,
  xleftmargin=2em,
  framexleftmargin=1.5em,
}
\title{The IsalProgram Programming Language}

\author{ \href{https://orcid.org/0000-0001-8231-5687}{\includegraphics[scale=0.06]{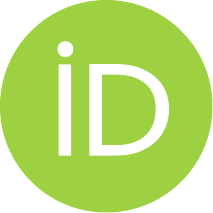}\hspace{1mm}Ezequiel L\'opez-Rubio}\thanks{Corresponding author. ITIS Software. Universidad de M\'alaga. C/ Arquitecto Francisco Peñalosa 18, 29010, Málaga, Spain} \\
	Department of Computer Languages and Computer Science\\
    University of M\'alaga\\
    Bulevar Louis Pasteur, 35\\
    29071 M\'alaga, Spain \\
	\texttt{ezeqlr@lcc.uma.es} \\
}

\renewcommand{\shorttitle}{The IsalProgram Programming Language}

\hypersetup{
  pdftitle  = {The IsalProgram Programming Language},
  pdfsubject= {cs.PL, cs.FL, cs.AI},
  pdfauthor = {Ezequiel L\'opez-Rubio},
  pdfkeywords={programming language, regular language, formal language theory,
               circular doubly linked list, program synthesis, neural transformers,
               Levenshtein distance}
}

\begin{document}
\maketitle

% ------------------------------------------------------------------
\begin{abstract}
We introduce \emph{IsalProgram} (Instruction Set and Language for
Programming), a novel assembly-like programming language with three
distinctive theoretical properties: (1)~it is a \emph{regular language}
in the sense of formal language theory, meaning its programs are
accepted by a finite automaton; (2)~every finite string over the
instruction alphabet is a syntactically valid program; and (3)~it
makes no explicit use of memory addresses or variable names, absolute or
relative.
Programs are finite sequences of tokens drawn from a fixed instruction
set, and are executed on a virtual machine whose sole data structure is
a circular doubly linked list (CDLL) navigated by three data pointers,
with control flow governed by two code pointers.
We give a complete formal definition of the language and its virtual
machine, prove its regularity, and demonstrate its expressive power.
We further discuss IsalProgram's potential advantages as a target
language for neural program synthesis, the amenability of its program
space to metric-based exploration via the Levenshtein edit distance,
and directions for analysing computability and complexity within this
framework.
\end{abstract}


% ==================================================================
\section{Introduction}
\label{sec:introduction}

Most practical programming languages occupy the upper tiers of the
Chomsky hierarchy~\citep{hopcroft2006automata}: they require context-free
or even context-sensitive grammars, and their parsers must maintain a
non-trivial state.
This is a reasonable price to pay for human readability, but it imposes
costs when languages are used as \emph{objects of study} rather than
as tools for human programmers—costs that manifest in program synthesis,
automated program repair, and the formal analysis of computability.

We present \textbf{IsalProgram}, a language that sits at the computational simplicity
end of the human readability/computational simplicity spectrum without sacrificing
computational power.
IsalProgram is:

\begin{itemize}
  \item \textbf{Regular.}  Its set of valid programs is a regular language.
        Any valid program is accepted by a trivial one-state deterministic finite
        automaton (DFA) after a single left-to-right scan.
  \item \textbf{Total.}  Every finite string over the instruction alphabet
        is a syntactically valid program.  There are no parse errors.
  \item \textbf{Address-free.}  No variable names, constant memory addresses,
        labels, or numeric literals are embedded in the instruction
        stream.  All data live in a circular doubly linked list (CDLL)
        navigated by three pointer registers.
\end{itemize}

Despite its syntactic minimality, IsalProgram is computationally
universal: we show in Section~\ref{sec:universality} that it can
simulate a Turing machine, placing it in the class of
Turing-complete languages~\citep{turing1936computable,sipser2012theory}.

The paper is organised as follows.
Section~\ref{sec:related} situates IsalProgram relative to existing
work on minimal and esoteric languages, neural program synthesis, and
program metrics.
Section~\ref{sec:language} gives the formal definition of the
instruction set.
Section~\ref{sec:vm} defines the virtual machine.
Section~\ref{sec:regularity} proves regularity.
Section~\ref{sec:universality} establishes Turing completeness.
Section~\ref{sec:functions} explains how to develop functions in IsalProgram.
Section~\ref{sec:discussion} discusses theoretical and practical
implications, and Section~\ref{sec:conclusion} concludes.

% ==================================================================
\section{Related Work}
\label{sec:related}

\paragraph{Minimal and esoteric languages.}
The closest predecessor in spirit is Brainfuck~\citep{brainfuck1993},
a Turing-complete language with only eight instructions operating on a
one-dimensional tape.
Like IsalProgram, Brainfuck is Turing-complete and address-free.
Unlike IsalProgram, its set of syntactically valid programs is
\emph{not} a regular language because matching bracket pairs requires
a pushdown automaton~\citep{sipser2012theory}.
IsalProgram removes this restriction entirely: its instruction set is
flat and bracket-free, making syntax trivially checkable. The instruction set of IsalProgram is much more powerful, so typical assembly-language-level calculations can be easily performed. Moreover, IsalProgram supports data types and modularity via function calls.

\paragraph{Assembly and low-level languages.}
Traditional assembly languages~\citep{hyde2010art,patterson2021computer}
expose hardware registers and memory addresses directly.
They require a two-pass assembler to resolve labels and are not regular.
IsalProgram replaces the address space with a dynamically-sized CDLL,
eliminating the need for symbolic addressing.

\paragraph{Neural program synthesis.}
Recent large language models have demonstrated impressive code
generation~\citep{chen2021evaluating,austin2021program} on languages
such as Python and Java.
However, the irregular syntax of such languages requires the model to
maintain complex structural constraints.
A regular target language removes this burden: any token sequence is
valid, so the model need not learn to balance brackets, match types, or
respect scoping rules~\citep{vaswani2017attention}.
IsalProgram has been designed with this use-case in mind.

\paragraph{Metric spaces of programs.}
The Levenshtein edit distance~\citep{levenshtein1966binary} between
two strings is a natural measure of proximity in program space.
For most languages, syntactic validity is not preserved under arbitrary
single-token edits, so the Levenshtein ball around a valid program often
contains many invalid programs.
Because IsalProgram is total, \emph{every} string in the Levenshtein
ball of any IsalProgram program is itself a valid program—a property
exploited in evolutionary and distance-based program
search~\citep{koza1992genetic,navarro2001guided}.

% ==================================================================
\section{The IsalProgram Language}
\label{sec:language}

\subsection{Alphabet and Programs}

Let $\Sigma$ denote the finite set of IsalProgram \emph{tokens}
(instructions).  A program is any element of $\Sigma^{*}$
(including the empty program).

\begin{definition}[IsalProgram program]
An \emph{IsalProgram program} is a finite sequence
$\pi = t_1\, t_2\, \ldots\, t_n$ with each $t_i \in \Sigma$.
The set of all programs is $\mathcal{L} = \Sigma^{*}$.
\end{definition}

\subsection{Instruction Set}
\label{subsec:instructions}

Table~\ref{tab:instructions} lists all instructions together with
their mnemonics.  They are grouped into seven categories:
flow control, pointer movement (code and data), node operations,
copy, arithmetic, string, and load.

\begin{table}[t]
\caption{IsalProgram instruction set.  Here $x,y \in \{p,s,t\}$
         (primary, secondary, ternary) with $x \neq y$.}
\label{tab:instructions}
\centering
\small
\begin{tabular}{@{}llp{7.2cm}@{}}
\toprule
\textbf{Category} & \textbf{Mnemonic} & \textbf{Description} \\
\midrule
\multirow{4}{*}{Flow control}
  & \texttt{J}     & Jump: set IP $\leftarrow$ JP \\
  & \texttt{Bp/Bs/Bt} & Branch if datum at primary/secondary/ternary pointer
                         is non-zero (true); set IP $\leftarrow$ JP \\
  & \texttt{Kp/Ks/Kt}     & Function call: push IP onto the stack and jump to the address stored at primary/secondary/ternary pointer \\
  & \texttt{R}     & Function return: pop the return address from the top of the stack to IP \\
  & \texttt{H}     & Halt: stop execution \\
  & \texttt{W}     & Wait: no operation \\
\midrule
\multirow{6}{*}{Pointer movement}
  & \texttt{M$xy$} & Move data pointer $x$ to position of data pointer $y$ \\
  & \texttt{Mji}   & Move JP to current IP position \\
  & \texttt{Np/Ns/Nt} & Advance primary/secondary/ternary pointer to next CDLL node \\
  & \texttt{Pp/Ps/Pt} & Retreat primary/secondary/ternary pointer to previous CDLL node \\
  & \texttt{Nj}    & Advance JP to next instruction \\
  & \texttt{Pj}    & Retreat JP to previous instruction \\
\midrule
\multirow{2}{*}{Node operations}
  & \texttt{Ib/Ii/If/Is} & Insert new Boolean/integer/float/string node after primary; advance primary \\
  & \texttt{D}     & Delete node at primary (no-op if CDLL has one node); advance pointers that pointed to the deleted node \\
\midrule
Copy
  & \texttt{C$xy$} & Copy datum from node at $x$ into node at $y$ \\
  & \texttt{C$jx$} & Copy JP into node at $x$ \\
  & \texttt{C$xj$} & If node at $x$ contains a valid address then copy it into JP\\
\midrule
\multirow{3}{*}{Arithmetic}
  & \texttt{Aa/As/Am/Ad} & Add/subtract/multiply/divide data at secondary and ternary; store in primary \\
  & \texttt{An}    & Negate datum at secondary; store in primary \\
  & \texttt{Aq}    & Square root of datum at secondary; store in primary \\
\midrule
String
  & \texttt{Sc}    & Concatenate strings at secondary and ternary; store in primary \\
  & \texttt{Sx}    & Extract substring of primary using indices at secondary (start) and ternary (end) \\
\midrule
\multirow{3}{*}{Load}
  & \texttt{Zp/Zs/Zt} & Zero datum at primary/secondary/ternary \\
  & \texttt{L1}--\texttt{L15} & Load integer constant 1–15 into node at primary \\
  & \texttt{Le}/\texttt{Lp} & Load Euler's number $e$ / $\pi$ into node at primary \\
\bottomrule
\end{tabular}
\end{table}

\subsection{Semantics of Error-Prone Instructions}

Several instructions are defined to be \emph{no-ops} on error rather
than raising exceptions.  Specifically:

\begin{itemize}
  \item \texttt{Ad}: if the divisor (ternary node) is zero, do nothing.
  \item \texttt{Aq}: if the argument is negative, do nothing.
  \item \texttt{Sc}: if either operand is not a string, or the primary
        node is not a string, do nothing.
  \item \texttt{Sx}: if primary is not a string or secondary/ternary are
        not integers, do nothing; negative indices are treated as in Python
        (counting from the end).
  \item \texttt{C$xy$}: types must be compatible; otherwise do nothing.
\end{itemize}

This design keeps every program semantically well-defined for all inputs,
reinforcing the totality property.

% ==================================================================
\section{The IsalProgram Virtual Machine}
\label{sec:vm}

\subsection{State}

A virtual machine (VM) state is a tuple
$\sigma = (\mathit{IP}, \mathit{JP}, \mathit{ST}, \mathit{CDLL}, p, s, t)$ where:

\begin{itemize}
  \item $\mathit{IP} \in \{1,\ldots,n+1,\mathsf{halt}\}$ is the
        \emph{instruction pointer}, with $n$ the length of the program.
  \item $\mathit{JP} \in \{1,\ldots,n+1,\mathsf{halt}\}$ is the \emph{jump pointer}.
  \item $\mathit{ST}$ is a stack of instruction positions.
  \item $\mathit{CDLL}$ is a non-empty circular doubly linked list of
        typed nodes.  Each node $v$ carries a datum
        $d(v) \in \mathbb{B} \cup \mathbb{Z} \cup \mathbb{R} \cup \Sigma^{*}$
        and a type $\tau(v) \in \{\texttt{bool},\texttt{int},
        \texttt{float},\texttt{string}\}$.
  \item $p, s, t$ are the primary, secondary, and ternary data pointers,
        each pointing to a node in the CDLL.
\end{itemize}

\subsection{Initialisation}

Given a program $\pi = t_1\ldots t_n$ and a non-empty input list
$\langle v_1,\ldots,v_k \rangle$ of typed values with $k>0$:

\begin{enumerate}
  \item Set $\mathit{ST}$ to the empty stack.
  \item Construct a CDLL with nodes $m_1,\ldots,m_k$ in order, linked
        circularly.
  \item If the program is empty ($n=0$), then set $\mathit{IP} = \mathsf{halt}$, $\mathit{JP} = \mathsf{halt}$. Otherwise, set $\mathit{IP} = 1$, $\mathit{JP} = 1$.
  \item Set $p = s = t = m_1$.
\end{enumerate}

\subsection{Execution}

Execution proceeds as follows.  While $\mathit{IP} \neq \mathsf{halt}$:

\begin{enumerate}
  \item Fetch instruction $t_{\mathit{IP}}$.
  \item Execute the instruction (see Table~\ref{tab:instructions} and
        Section~\ref{subsec:instructions}).
  \item Unless the instruction explicitly modifies $\mathit{IP}$,
        set $\mathit{IP} \leftarrow \mathit{IP}+1$.
  \item If $\mathit{IP} > n$, execution terminates (implicit halt).
\end{enumerate}

The output of the program is the sequence of data values currently
stored in the CDLL at termination, starting from the primary pointer $p$ to the next nodes.

\subsection{The CDLL}

Figure~\ref{fig:cdll} illustrates the CDLL structure with three data
pointers and the relationship between node operations.

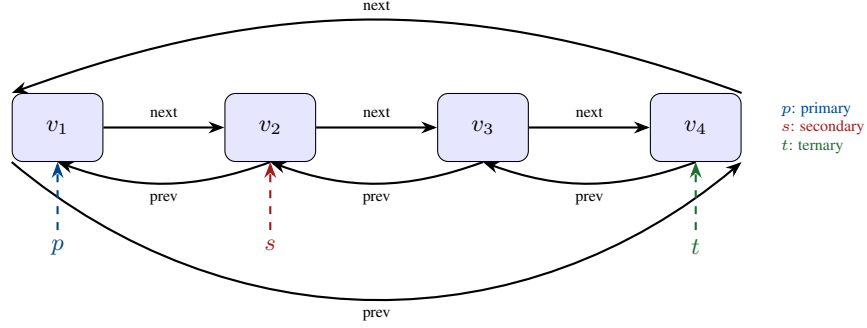
\begin{figure}[t]
\centering
\begin{tikzpicture}[
  node distance = 1.6cm,
  box/.style    = {draw, rounded corners=4pt, minimum width=1.2cm,
                   minimum height=0.9cm, font=\small\ttfamily},
  ptr/.style    = {-{Stealth[length=6pt]}, thick},
  circptr/.style= {-{Stealth[length=6pt]}, thick, bend left=30},
  dataptr/.style= {-{Stealth[length=6pt]}, dashed, thick, color=isalblue}
]

% CDLL nodes
\node[box, fill=blue!10]  (n1) {$v_1$};
\node[box, fill=blue!10, right=of n1]  (n2) {$v_2$};
\node[box, fill=blue!10, right=of n2]  (n3) {$v_3$};
\node[box, fill=blue!10, right=of n3]  (n4) {$v_4$};

% forward links
\draw[ptr] (n1.east) -- node[above, font=\tiny]{next}(n2.west);
\draw[ptr] (n2.east) -- node[above, font=\tiny]{next}(n3.west);
\draw[ptr] (n3.east) -- node[above, font=\tiny]{next}(n4.west);
\draw[ptr, bend right=20] (n4.north east) to node[above, font=\tiny]{next} (n1.north west);

% backward links
\draw[ptr, bend left=20] (n2.south) to node[below, font=\tiny]{prev} (n1.south);
\draw[ptr, bend left=20] (n3.south) to node[below, font=\tiny]{prev}(n2.south);
\draw[ptr, bend left=20] (n4.south) to node[below, font=\tiny]{prev}(n3.south);
\draw[ptr, bend right=40] (n1.south west) to node[below, font=\tiny]{prev} (n4.south east);

% data pointers
\node[font=\small\bfseries, color=isalblue, below=0.9cm of n1] (pp) {$p$};
\node[font=\small\bfseries, color=isalred,  below=0.9cm of n2] (sp) {$s$};
\node[font=\small\bfseries, color=isalgreen,below=0.9cm of n4] (tp) {$t$};

\draw[dataptr, color=isalblue]  (pp.north) -- (n1.south);
\draw[dataptr, color=isalred]   (sp.north) -- (n2.south);
\draw[dataptr, color=isalgreen] (tp.north) -- (n4.south);

% legend
\node[right=0.4cm of n4, font=\tiny, align=left] (leg) {
  \textcolor{isalblue}{$p$: primary}\\
  \textcolor{isalred}{$s$: secondary}\\
  \textcolor{isalgreen}{$t$: ternary}
};

\end{tikzpicture}
\caption{Circular doubly linked list (CDLL) with four nodes.
         The three data pointers ($p$, $s$, $t$) can point to arbitrary
         nodes independently.  The list wraps around: the successor of
         $v_4$ is $v_1$ and the predecessor of $v_1$ is $v_4$.}
\label{fig:cdll}
\end{figure}

\subsection{Pointer Architecture}

Figure~\ref{fig:pointers} shows the full pointer architecture of the VM,
including both the code pointers (IP and JP) and the
data pointers ($p$, $s$, $t$).

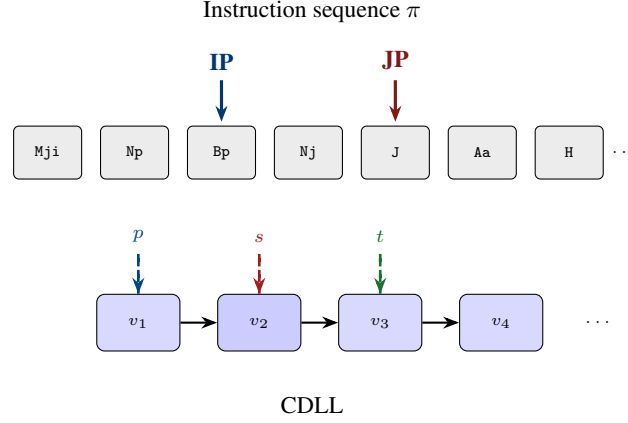
\begin{figure}[t]
\centering
\begin{tikzpicture}[
  insbox/.style = {draw, fill=gray!15, minimum width=0.9cm,
                   minimum height=0.7cm, font=\tiny\ttfamily, rounded corners=2pt},
  ptr/.style    = {-{Stealth[length=6pt]}, thick},
  label/.style  = {font=\small\bfseries}
]

% Instruction sequence
\node[font=\small, above] at (3.5,3.2) {Instruction sequence $\pi$};
\foreach \i/\tok in {0/Mji,1/Np,2/Bp,3/Nj,4/J,5/Aa,6/H}{
  \node[insbox] (i\i) at (\i*1.15, 1.6) {\tok};
}
\node[font=\tiny] at (7*1.15-0.4, 1.6) {$\cdots$};

% IP arrow
\draw[ptr, color=isalblue!80!black, line width=1.2pt]
  (i2.north) ++(0,0.60) 
  node[above, font=\small\bfseries, color=isalblue!80!black] {IP} -- ++(0,-0.55);

% JP arrow
\draw[ptr, color=isalred!80!black, line width=1.2pt]
  (i4.north) ++(0,0.60) 
  node[above, font=\small\bfseries, color=isalred!80!black] {JP} -- ++(0,-0.55);

% CDLL below
\node[font=\small, below] at (3.5,-1.5) {CDLL};
\foreach \j/\col in {0/blue!15,1/blue!20,2/blue!15,3/blue!15}{
  \node[draw, fill=\col, rounded corners=3pt,
        minimum width=1.1cm, minimum height=0.75cm,
        font=\tiny] (d\j) at (\j*1.6+1.2,-0.65) {$v_{\the\numexpr\j+1\relax}$};
}
\node[font=\tiny] at (4*1.6+1.2-0.3,-0.65) {$\cdots$};

\draw[ptr] (d0.east) -- (d1.west);
\draw[ptr] (d1.east) -- (d2.west);
\draw[ptr] (d2.east) -- (d3.west);

% data pointer arrows
\draw[ptr, dashed, color=isalblue, line width=1pt]
  (d0.north) -- ++(0,0.55)
  node[above, font=\scriptsize\bfseries, color=isalblue]{$p$} -- ++(0,-0.55);

\draw[ptr, dashed, color=isalred, line width=1pt]
  (d1.north) -- ++(0,0.55)
  node[above, font=\scriptsize\bfseries, color=isalred]{$s$} -- ++(0,-0.55);

\draw[ptr, dashed, color=isalgreen, line width=1pt]
  (d2.north) -- ++(0,0.55)
  node[above, font=\scriptsize\bfseries, color=isalgreen]{$t$} -- ++(0,-0.55);

\end{tikzpicture}
\caption{Full pointer architecture of the IsalProgram VM.
         Two instruction-level pointers (IP and JP) index into the
         program sequence; three data-level pointers ($p$, $s$, $t$)
         index into the CDLL.}
\label{fig:pointers}
\end{figure}

% ==================================================================
\section{Regularity of IsalProgram}
\label{sec:regularity}

\begin{theorem}
  \label{thm:regular}
  The set of IsalProgram programs $\mathcal{L} = \Sigma^{*}$ is a regular
  language.
\end{theorem}

\begin{proof}
  By definition, $\mathcal{L} = \Sigma^{*}$ is the Kleene closure of
  the finite alphabet~$\Sigma$~\citep{kleene1956representation}.
  The language $\Sigma^{*}$ is accepted by the one-state DFA
  $M = (\{q_0\}, \Sigma, \delta, q_0, \{q_0\})$ where
  $\delta(q_0, a) = q_0$ for all $a \in \Sigma$.
  Since $M$ is a DFA, $\mathcal{L}$ is regular.
\end{proof}

The key design decision that enables Theorem~\ref{thm:regular} is the
absence of any bracket-matching, label-resolution, or type-annotation
requirements: there is simply nothing to check.
Figure~\ref{fig:dfa} depicts the minimal DFA for $\mathcal{L}$.

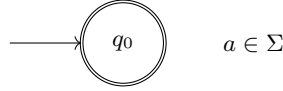
\begin{figure}[t]
\centering
\begin{tikzpicture}[
  state/.style={draw, circle, minimum size=1.1cm, font=\small},
  every edge/.style={-{Stealth[length=6pt]}, thick}
]
\node[state, double] (q0) {$q_0$};
\draw[->] (-1.5,0) -- (q0);
\path (q0) edge [loop right, looseness=8] node [right, font=\small]
      {$a \in \Sigma$} (q0);
\end{tikzpicture}
\caption{Minimal DFA for $\mathcal{L} = \Sigma^{*}$.
         The single state $q_0$ is both the start state and the unique
         accepting state.  Every token loops back to $q_0$.}
\label{fig:dfa}
\end{figure}

\begin{corollary}
  Every finite string over $\Sigma$ is a syntactically valid
  IsalProgram program.
\end{corollary}

\begin{corollary}
  The set of IsalProgram programs of length at most $n$ has cardinality
  $\sum_{k=0}^{n} |\Sigma|^k$.
\end{corollary}

% ==================================================================
\section{Computational Universality}
\label{sec:universality}

Although IsalProgram has a trivially simple syntax, its virtual machine
is computationally universal.

\begin{theorem}
  \label{thm:universal}
  IsalProgram is Turing-complete.
\end{theorem}

\begin{proof}[Proof sketch]
  We simulate a two-tape Turing machine (TM)~\citep{turing1936computable,
  minsky1967computation}.  The CDLL models both tapes concatenated:
  the primary pointer acts as the TM's head, and the secondary pointer
  marks the boundary between the two tapes.
  The ternary pointer holds the current state index as an integer.

  \emph{Tape operations.}
  \texttt{Np}/\texttt{Pp} move the head right/left.
  \texttt{Ii}/\texttt{If}/\texttt{Is} extend the tape on demand
  (simulating the infinite tape).
  \texttt{D} contracts the tape, which is not required to simulate the infinite tape but could be employed to reduce the number of nodes in the CDLL.

  \emph{State transitions.}
  Each TM transition
  $\delta(q, a) = (q', a', d)$
  is encoded as a fixed-length IsalProgram subprogram that
  (i) tests the symbol at the head using \texttt{Bp/Bs/Bt},
  (ii) writes the new symbol with a copy instruction, and
  (iii) updates the state register and moves the head.

  \emph{Control flow.}
  The \texttt{Kx/R} mechanism replaces the standard dispatch table:
  \texttt{Kx} saves the return address and performs the call; \texttt{R} returns from the call.  %A linear scan of state-handler subprograms replaces indirect

  Since every TM can be encoded in this way, and IsalProgram programs
  can be executed on the IsalProgram VM, IsalProgram is
  Turing-complete~\citep{sipser2012theory}.
\end{proof}

The tension between Theorem~\ref{thm:regular} and
Theorem~\ref{thm:universal} is only apparent: regularity is a property
of the \emph{syntax} (the set of strings), while Turing-completeness
is a property of the \emph{semantics} (what the programs compute).
As~\citet{brainfuck1993} demonstrated, a language with a trivially
simple syntax can still be Turing-complete.

\section{Functions}
\label{sec:functions}

IsalProgram supports a structured notion of \emph{functions} that
avoids any literal memory address constant in the code.
A function call works as follows:

\begin{enumerate}
  \item The caller ensures that the primary pointer and the
        subsequent CDLL nodes hold the function's input arguments.
  \item The function body inserts temporary nodes for local
        computation immediately after those input nodes.
  \item Upon return, the function deletes all temporary nodes and
        stores its results into the primary pointer and the
        subsequent CDLL nodes.
\end{enumerate}

Figure~\ref{fig:functioncall} illustrates the function call protocol. The call/return mechanism uses the \texttt{Kx/R} pair: \texttt{Kx} pushes the current IP to the stack and jumps to the function at the address stored in the node pointed by x, and a final \texttt{R} at the end of the
function body returns control to the address retrieved from the stack.
This is analogous to a \texttt{call}/\texttt{ret} pair in assembly,
without any function address specification in the code.

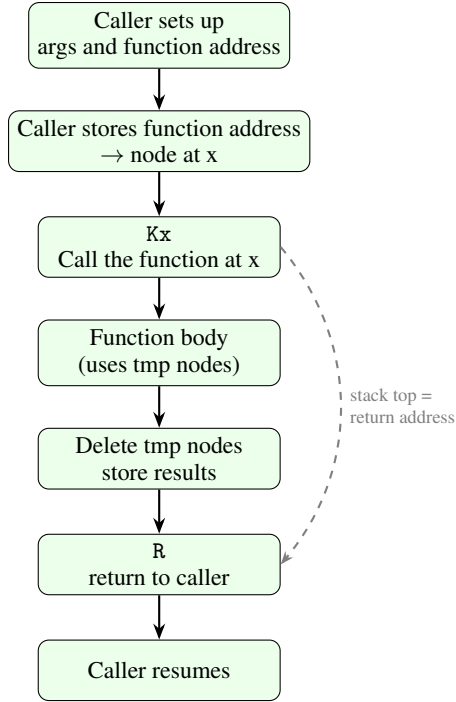
\begin{figure}[t]
\centering
\begin{tikzpicture}[
  block/.style={draw, fill=green!8, rounded corners=4pt,
                minimum width=3.2cm, minimum height=0.8cm,
                font=\small, align=center},
  arr/.style={-{Stealth[length=6pt]}, thick},
  dashed arr/.style={-{Stealth[length=5pt]}, dashed, thick, gray}
]

\node[block] (caller) at (0,0)    {Caller sets up \\ args and function address};
\node[block] (mji)    at (0,-1.4) {Caller stores function address \\ $\to$ node at x};
\node[block] (call) at (0,-2.8) {\texttt{Kx}\\Call the function at x};
\node[block] (body)   at (0,-4.2) {Function body\\(uses tmp nodes)};
\node[block] (clean)  at (0,-5.6) {Delete tmp nodes\\store results};
\node[block] (ret)    at (0,-7.0) {\texttt{R}\\return to caller};
\node[block] (resume) at (0,-8.4) {Caller resumes};

\draw[arr] (caller) -- (mji);
\draw[arr] (mji)    -- (call);
\draw[arr] (call)    -- (body);
\draw[arr] (body)   -- (clean);
\draw[arr] (clean)  -- (ret);
\draw[arr] (ret)    -- (resume);

% side annotation
\draw[dashed arr, bend left=40]
  (call.east) to node[right, font=\scriptsize, align=left]{stack top =\\ return address} (ret.east);

\end{tikzpicture}
\caption{Function call protocol in IsalProgram.
         \texttt{Kx} saves the return address into the stack, and jumps to the function whose address is stored at the node pointed by x. After the function has completed its calculations, a final \texttt{R} at the end of the function body returns control to the caller.}
\label{fig:functioncall}
\end{figure}

% ==================================================================
\section{Discussion}
\label{sec:discussion}

\subsection{IsalProgram as a Target for Neural Program Synthesis}

Recent work has shown that large language models (LLMs) trained on
human-written code~\citep{chen2021evaluating,austin2021program} can
synthesise programs from natural-language descriptions.
However, all languages studied so far have complex, non-regular syntax:
the model must simultaneously learn both the \emph{meaning} and the
\emph{form} of programs.

IsalProgram separates these concerns.
Because $\mathcal{L} = \Sigma^{*}$, the model's output is \emph{always}
syntactically valid; no post-generation repair or filtering is needed.
The model can concentrate entirely on the semantics—the
\emph{what-to-compute} problem.

Formally, for a model parameterised by $\theta$ that generates a
sequence of tokens $t_1, \ldots, t_n$ via
\begin{equation}
  p_\theta(\pi) = \prod_{i=1}^{n} p_\theta(t_i \mid t_1,\ldots,t_{i-1}),
\end{equation}
validity is guaranteed for all $\theta$ and all $n$, which simplifies
training objectives and evaluation~\citep{vaswani2017attention}.

\subsection{The Levenshtein Metric on Program Space}

The edit distance~\citep{levenshtein1966binary} between two programs
$\pi, \pi'$ is
\begin{equation}
  d(\pi, \pi') = \text{ed}(\pi, \pi'),
\end{equation}
the minimum number of single-token insertions, deletions, and
substitutions needed to transform $\pi$ into $\pi'$.  The totality of IsalProgram means the shortest path between two programs according to the edit distance is composed of valid programs.

This makes IsalProgram amenable to nearest-neighbour program search,
genetic programming~\citep{koza1992genetic}, and mutation-based automatic program repair: any mutant is a valid program.

\subsection{Computability and Complexity Analysis}

Standard computability arguments are developed with respect to Turing
machines~\citep{turing1936computable,sipser2012theory}.
The IsalProgram VM provides an alternative, arguably more natural,
model: programs are finite strings rather than state-transition tables,
and the memory model (CDLL) is more structured than a tape but more
flexible than a fixed-size register file.

We conjecture that the standard time and space complexity classes
(PTIME, PSPACE, etc.~\citep{papadimitriou1994computational}) can be
characterised cleanly in terms of the number of CDLL operations and
the maximum CDLL size, respectively.
Formalising this connection is left as future work.

\subsection{Limitations}

\paragraph{Three-pointer constraint.}
The three data
pointers impose a strict register-pressure cost: some computations
require careful sequencing of pointer-move instructions that would be
trivial with four or more pointers.
Extending the VM to four or more pointers while preserving regularity is
straightforward (since regularity depends only on the absence of
syntax constraints, not on the number of pointers).

\paragraph{No immediate literals.}
IsalProgram can load only the constants $0$, $1$--$15$, $e$, and $\pi$
directly.
Larger constants must be composed arithmetically.
This is a feature—it keeps the alphabet finite—but it makes some
programs verbose.

\paragraph{Halting problem.}
As with all Turing-complete languages, the halting problem for
IsalProgram is undecidable~\citep{turing1936computable}.
%The implicit halt at the end of the instruction sequence provides a natural timeout for programs of bounded length.

% ==================================================================
\section{Conclusion}
\label{sec:conclusion}

We have introduced IsalProgram, a novel programming language with three
remarkable properties that coexist in a single
design: (1)~its syntax is a regular language (in fact $\Sigma^*$),
(2)~every string is a valid program, and (3)~it is Turing-complete
despite using no variable names or literal memory addresses.

We have given a precise formal specification of the instruction set and
virtual machine, proved regularity and Turing-completeness, and
developed a function call protocol.

The language opens three concrete directions for future research.
First, the use of IsalProgram as a target language for neural program
synthesis, where the totality of the language eliminates the syntactic
validity constraint on generated sequences.
Second, the exploitation of the Levenshtein metric over $\Sigma^*$ for
evolutionary or distance-based program search.
Third, the development of a complexity theory grounded in CDLL
operations rather than Turing machine steps.

We believe IsalProgram offers a uniquely clean theoretical laboratory
for studying the relationship between syntax, semantics, and
computation.

% ==================================================================
%\section*{Acknowledgements}

%The authors thank the reviewers for their constructive feedback.

% ==================================================================
\bibliographystyle{unsrtnat}
\bibliography{references}

\end{document}